\documentclass[10pt]{article}
\usepackage{anysize,graphicx}
\usepackage{hyperref}
\usepackage{amsmath,amsthm,amsfonts}
\usepackage{caption}

\newtheorem{definition}{Definition}
\newtheorem{theorem}{Theorem}

\begin{document}
\title{The $k$-path coloring problem in graphs with bounded treewidth: an application in integrated circuit manufacturing}

%
%
%
%

\author{Dehia Ait-Ferhat, Vincent Juliard, Gautier Stauffer, Andres Torres}

%
%
%
%
%


\maketitle

\begin{abstract}

In this paper, we investigate the $k$-path coloring problem, a variant of vertex coloring arising in the context of integrated circuit manufacturing. In this setting, typical industrial instances exhibit a `tree-like' structure. We exploit this property to design an efficient algorithm for our industrial problem: (i) on the methodological side, we show that the $k$-path coloring problem can be solved in polynomial time on graphs with bounded treewidth and we devise a simple polytime dynamic programming algorithm in this case (not relying on Courcelle's celebrated theorem); and (ii) on the empirical side, we provide computational evidences that the corresponding algorithm could be suitable for practice, by testing our algorithm on true instances obtained from an on-going collaboration with Mentor Graphics. We finally compare this approach with integer programming on some pseudo-industrial instances. It suggests that dynamic programming cannot compete with integer programming when the tree-width is greater than three. While all our industrial instances exhibit such a small tree-width, this is not for granted that all future instances will also do, and this tend to advocate for integer programming approaches.

\end{abstract}

\section{Introduction}\label{introduction}

Integrated circuits are made of several layers. In a nutshell,  the bottom layers contains the transistors, while the other layers (called {\em metal layers}) are used to connect the different components to comply with the designed functionalities of the device. We usually distinguish two types of components in metal layers: {\em vias} and {\em `wires'}. The former components are used for vertical connections and the latter for horizontal ones, see Fig. \ref{integrated_circuit} for an illustration.

	\begin{center}
	\includegraphics[scale = 0.25]{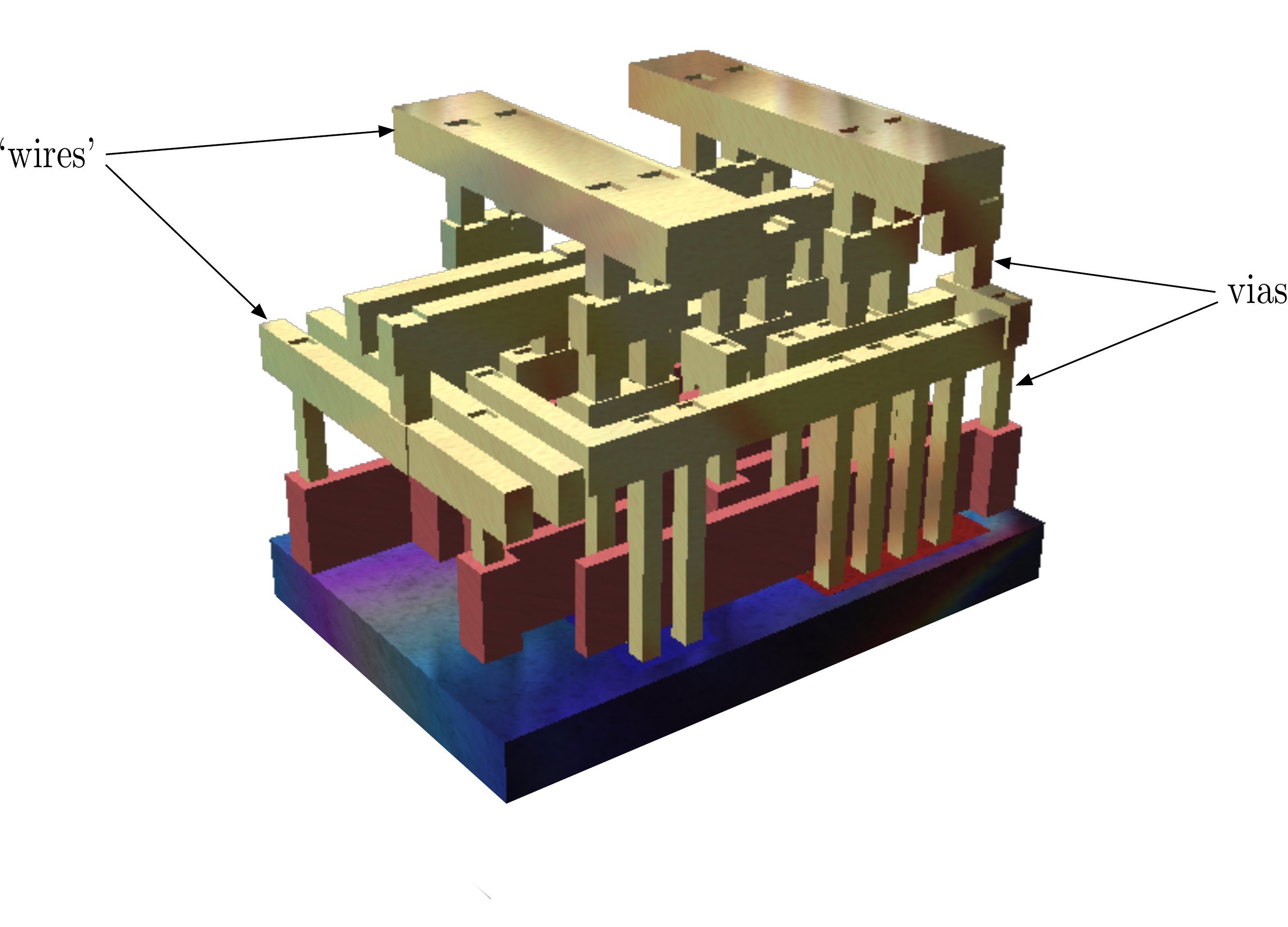} 
	\captionof{figure}{a 3D view of an integrated circuit (source: \url{https://commons.wikimedia.org/wiki/File:Silicon_chip_3d.png})}\label{integrated_circuit}
	\end{center}

Integrated circuit manufacturing involves the production of each layer of the circuit iteratively. The core technique used in  production is {\em lithography}. In brief, the idea is to etch each layer by exposing a photosensitive material to a light source through a mask: this creates a kind of {\em mould}, which is then filled with an appropriate conductor material (the mould is later removed through some chemical process).  Optical distortion might however result in the fusion of components if attention is not paid to keep the minimum distance between any two components above a certain threshold,  called {\em the lithography distance}. When some components are below this distance, the production of the mould is typically decomposed into several rounds of lithography: sub-masks are defined for (feasible) subsets of components, and the mould is produced in sequence, a process called {\em multiple patterning}. The process requires the proper alignment of the sub-masks which is a major challenge as the number of rounds increases. Hence multiple patterning induces additional costs and time in the production process and the industry tries to keep the number of patterning steps small. The problem of finding the minimum number of patterning steps readily translates into a vertex coloring problem. Indeed, one can build a graph $G$ whose node set is the set of components and two components are adjacent if they are at a distance less or equal to the lithography distance. This graph is sometimes called the {\em conflict graph}. The minimum number of rounds needed to produce the mould is the chromatic number of $G$. 

DSA-aware Multiple Patterning (DSA-MP) is a technique that combines Directed Self-Assembly (DSA) technology with lithography in order to go beyond the resolution limit imposed by lithography alone in the fabrication of integrated circuits.  Due to the nature of the process, it is particularly useful for the manufacturing of {vias}.  DSA-MP allows to introduce within a same sub-mask some vias that are closer to the lithography distance, under certain conditions.  The main idea is to intentionally fuse some vias within larger objects, called {\em guiding patterns}, and then to correct the corresponding design with DSA in a second step. More precisely, the mould, possibly obtained after applying several rounds of lithography, is filled with a {\em block copolymer} in a random state that self-organize in a structured way after a chemical reaction is triggered: if the guiding patterns are designed appropriately, and thus if the vias that are fused satisfy some specific properties, one of the two polymers assembles into a set of cylinders that can then be used for via creation after some additional processing.  We refer the reader to \cite{IP} for more details about integrated circuit manufacturing and DSA-aware multiple patterning.

One of the core problem in DSA-MP is again the problem of minimizing the number of lithography steps as the production time and costs are still dominated by lithography. We showed in \cite{IP} that several variants are of interest in the industry. One such variant consists in fusing ``small chains of vias'' and it reduces to the following variant of vertex coloring: given a graph $G=(V,E)$, a subset $F$ of edges of $E$ and an integer $k\geq 0$, color the vertices of $G$ with a minimum number of colors so that each color induces, in $G$, a disjoint union of paths of $G_F=(V,F)$ of length less or equal to $k$. We call the problem the {\em $k$-path coloring problem}. Concretely, the node set of $G$ is the set of vias in a layer, the edges correspond to the pairs of vias that are closer than the lithography distance, $F$ are the pairs of vias whose distance is within a certain range specific to the DSA technology used (but smaller or equal to the lithography distance), and $k$ is a small number. The Electronic Design Automation (EDA) industry is especially interested in the $1$-path and the $2$-path coloring problems as they correspond to the current technological capabilities.  Regular (vertex) coloring is equivalent to $0$-path coloring, which proves that (i) $k$-path coloring is NP-hard in general and (ii) that we can only improve upon standard multiple patterning by choosing $k\geq 1$.  The problem was introduced under the same name\footnote{Some authors use the same terminology for another variant of graph coloring, see for instance \cite{Frick,Johns,Mynhardt}.} in \cite{Akiyama} in the special case where $F=E$. In particular it is known that $k$-path $L$-colorability is already hard for $L=2$ and $k=1$ \cite{Jinjiang1} and for $L=3$ and $k=2$ \cite{Jinjiang2}.

Mentor graphics has developed in-house heuristics to solve the problem quickly. A typical design can be made of several hundred thousands of vias per layer and Mentor Graphics' heuristics can solve (approximately) these instances in less than a second.  However the company has interest in fast exact algorithms as well. Initially, their main interest for exact approaches  lay in  the possibility to assess the quality of their heuristics. However, because of our first encouraging results with integer programming  \cite{IP}, the company understood that exact approaches might actually find their way to production. While testing different integer programming models in \cite{IP}, we observed that  typical industrial instances exhibit some structure. We indeed noticed that most instances are extremely sparse and typically `tree-like'. The main reason is that the design of the circuit (and the placement of vias in particular) is made with lithography constraints in mind so that distances between vias is kept as large as possible. In this project, we decided to test whether exact algorithms exploiting the `tree-like' property could be competitive with heuristics (computationally wise) for production. A well-known mesure of ``tree-likeness'' is the notion of treewidth introduced by Robertson and Seymour \cite{Robertson_seymour_1986}. 

\begin{definition}{(Tree decomposition ~\cite{Robertson_seymour_1986}).}
Let $G=(V,E)$ be a graph. A {\em tree decomposition} of $G$ is a pair $(X,T)$ where $T$ is a tree and $X=\{X_1, \dots, X_n\}$ is the set of nodes of $T$, called {\em bags} such that $\forall i \in \{1, \dots, n\}$, $X_i \subseteq V$ and the three following conditions are verified:

\begin{enumerate}

	\item $\cup_{i \in \{1, \dots, n\}} X_i = V$.

	\item  $\forall (u,v) \in E, \exists X_i \in X$ such that $u,v \in X_i$.

	\item $\forall u \in V$, the bags containing $u$ is a connected sub-tree of $T$.
\end{enumerate}
\end{definition}

The {\em width} of a tree decomposition is the size of the largest bag minus one. The {\em treewidth} of a graph $G$ is the minimum width over all possible tree decomposition of $G$. In particular the treewidth of a tree is $1$ (each edge can be used as a bag).

There are many combinatorial optimization problems that are hard in general but that can be solved in polynomial time on graphs with bounded treewidth (see for instance \cite{parametrizedcomplexity}). Such problems include stable set, dominating set, vertex coloring, Steiner tree, feedback vertex set, hamiltonian path, etc...  Besides, Courcelle~\cite{Courcelle_1990} has shown that a large class of problems on graphs can be solved in linear time when restricted to graphs with bounded treewidth. Courcelle's theorem essentially states that every graph property that can be formulated in {\em monadic second order logic} (MSOL)\footnote{Monadic second order logic in graphs is a formulation of a property of a graph using logical connectors ($\land$, $\lor$ $\lnot$, $\iff$, etc...), quantification on vertices and sets of vertices ($\forall v \in V$, $\exists v \in V$, $\forall V' \subseteq V$, $\exists V' \subseteq V$), quantification on edges and sets of edges ($\forall e \in E$, $\exists e \in E$, $\forall E' \subseteq E$, $\exists E' \subseteq E$), membership tests ($e\in F$, $v\in W$, etc...) and incidence tests ($v$ endpoint of $e$, $(u,v)\in E$).}  can be decided in linear time when restricted to graphs with bounded treewidth.  It is possible to prove that $k$-path $L$-colorability falls under Courcelle's theorem umbrella, see  \cite{Dehia}. It then follows that $k$-path coloring is polynomial in graphs with bounded treewidth \footnote{It is part of folklore that the chromatic number of a graph with treewidth $w$ is at most $w+1$ (we can `greedily' color the vertices as there is always a vertex of degree less or equal to $w$, see for instance \cite{Bodlaender2}). Hence the $k$-path chromatic number is also bounded by $w+1$. We can in particular restrict testing for $k$-path $L$-colorability to $L=1,...,w+1$.}.  Unfortunately Courcelle's theorem, albeit linear in the graph size, is considered impractical \cite{parametrizedcomplexity}: {\em ``Courcelle's theorem and its
variants should be regarded primarily as classification tools, whereas designing efficient dynamic-programming routines on tree decompositions requires
`getting your hands dirty' and constructing the algorithm explicitly.''}  In this paper, we develop such a direct dynamic programming approach for the $k$-path $L$-coloring problem and we test the performances of the corresponding algorithm on real-world instances from Mentor Graphics arising from DSA-aware Multiple Patterning for $k=1$ and $k=2$ (the cases of interest in the industry).

\subsection*{Additional definitions and notations}

Given a graph $G=(V,E)$ and a subset $E'$ of edges of $E$, we call a $E'$-neighbor of a vertex $v$, a vertex $u$ such that $(u,v)\in E'$.  Also a $E'$-path denotes a path with edges in $E'$. For a graph $G$, we sometimes denote by $V(G)$ the set of vertices of $G$, by $E(G)$ the set of edges of $G$, and, for $U\subseteq V(G)$, by $E[U]$ the set of edges of $G$ induced by the vertices in $U$ (that is with both extremities in $U$). 

We call a pair $(G,f)$, where $G=(V,E)$ is a graph and $f$ is a positive weight function on its edge set, a {\em weighted graph}. We can represent a weighted graph by its {\em (weighted) adjacency matrix}, that is a matrix $A$ of size $|V|\times |V|$ such that $A(u,v)=f((u,v))>0$ for all $u,v\in V: (u,v)\in E$ and $A(u,v)=0$ for all $u,v\in V: (u,v)\not\in E$. Note that the {\em support} of $A$, that is, the 0/1 matrix of same dimension as $A$ whose ones indicate pairs $(u,v)$ for which $A(u,v)>0$, is the (standard) adjacency matrix of $G$. A graph $G$ can be considered as a weighted graph with weight function $f=\bf 1$. 

Let $P$ be a path of $G$ with node set $v_0,...,v_{p}$, for some integer $p\geq 0$, and edge set $\{(v_i,v_{i+1})$ for $i=0,...,p-1\}$ and let $U\subseteq V$. The vertices of $P$ with smallest and largest index $i$ such that $v_i\in U$ are called the $U$-extremities of $P$. If $P$ has only one node in $U$, we call the corresponding node a $U$-two-extremity. By opposition, we call $v_0,v_p$ the {\em true} extremities of $P$. We call a node of $P$ {\em internal} if it is not a true extremity. 
Given a path $P$ such that $V(P)\cap U\neq \emptyset$ and a function $f:E(P)\mapsto {\mathbb R}^+ \setminus \{0\}$ , we define the {\em trace on $U$} of the weighted path $(P,f)$ as the weighted graph obtained from $P$ by `shrinking' the internal nodes of $P$ not in $U$. More formally, if $i_1<...<i_{l}$ are the indices of the vertices of $U$ on $P$, for some integer $l\geq 1$, the trace of $(P,f)$ on $U$ is the weighted graph $(P',f')$, where $P'$ is the path with vertex set $\{v_0,v_{i_1},...,v_{i_{l}},v_{p}\}$ and edge set $\{(v_0,v_{i_1}),(v_{i_{l}},v_{p})\} \cup \{(v_{i_j},v_{i_{j+1}})$ for $j=1,...,l-1\}$,  and for any edge $e$ of $P'$, $f'(e)$ is the length with respect to $f$ of the (sub)path of $P$ in-between the two end points of $e$. We define the trace on $U$ of a union of disjoint paths as the union of the traces of each path intersecting $U$.


A {\em rooted tree decomposition} is a tree decomposition where a bag is chosen as a root and the edges of the tree are oriented in the direction of the root bag. In a rooted tree decomposition $(X,T)$, we can naturally define {\em children} and {\em parent} bags and then, for a bag $X_i$, we denote by $V_i$  the union of the bags in the subtree of $T$ rooted at $X_i$, and by $G_i$ the subgraph of $G$ induced by the nodes in $V_i$.


\section{Dynamic programming}

We now develop a dynamic programming algorithm to solve the $k$-path $L$-coloring problem on graphs with bounded treewidth. It is convenient to present the algorithm on a {\em nice tree decomposition}. 


\begin{definition}{(Nice tree decomposition).}
Let $G=(V,E)$ be a graph. A nice tree decomposition is a rooted tree decomposition of $G$ where every bag  $X_i$ of the tree has at most two children and is one of the four following types:
 \begin{itemize}
 	\item \textbf{Leaf bag:}  $X_i$ has no children and contains only one vertex $v \in V$, i.e. $|X_i|=1$.
 	\item \textbf{Introduce bag:} $X_i$ has exactly one child bag noted $X_j$ such that $X_i = X_j \cup \{v\}$ for some $v \in V$ (and $v \not \in V_j$).
 	\item \textbf{Forget bag:} $X_i$ has exactly one child bag noted $X_j$ such that $X_i = X_j \backslash \{v\}$ for some $v\in X_j$.
 	\item \textbf{Join bag:} $X_i$ has exactly two children noted $X_{j_1}$ and $X_{j_2}$ such that $X_i = X_{j_1} = X_{j_2}$.
 \end{itemize}
\end{definition}

 Kloks \cite{Kloks_1994} proved that a tree decomposition can be converted into a nice tree decomposition (with at most four times the number of vertices in $G$) in linear time, while preserving the same treewidth. We assume thus that we are given such a nice tree decomposition.

\begin{theorem} There exists an algorithm that, given a graph $G=(V,E)$, a set of edges $F \subseteq E$, and a nice tree decomposition $(X,T)$ of $G$ of width $w$ solves the $k$-path $L$-coloring problem in time $O(L^w.k^{2(3w)^2}.(3w)^2.|T|)$.
\end{theorem}
\begin{proof}
Let $G=(V,E)$ be a graph, let $F \subseteq E$, and let us consider a nice tree decomposition $(X,T)$ of $G$ of width $w$. The general idea behind a dynamic programming approach to a decision problem on graph with bounded treewidth is to evaluate in each bag $X_i$ whether there is a solution to the problem restricted to $G_i$, and to store enough information about the corresponding solutions, to propagate and extend the solutions from the children bag(s) to the parent node iteratively. Hence, in our setting, for each bag $X_i$, we would like to know if there exists a $k$-path $L$-coloring of $G_i$. A $k$-path $L$-coloring of $G_i$ is obviously a $k$-path $L$-coloring of its children bag(s). Hence we can try to identify the $k$-path $L$-coloring of $G_i$ by checking which $k$-path $L$-coloring of the children bag(s) can be extended to a $k$-path $L$-coloring of $G_i$. Now, because we want to design an efficient algorithm for testing $k$-path $L$-colorability of $G$, we cannot keep a full list of all $k$-path $L$-colorings of $G_j$ for all bags $X_j$ as the list could grow exponentially as we move up the tree. Because each color of a $k$-path $L$-coloring induces a disjoint union of $F$-paths, and because of the structure of a tree decomposition, it is enough, as we will discuss in detail below, to keep the colors of the vertices in $X_i$ and the trace on $X_i$ of the $F$-paths in each color.  We call such a solution a {\em partial solution\footnote{This is the standard terminology in the field.}} as it can be extended to build a $k$-path $L$-coloring of $G_i$. Because we can encode, for each bag $X_i$, the (disjoint) union of the traces on $X_i$ of the $F$-paths of each color as a weighted graph with at most $3w$ vertices\footnote{Each path in the trace has at least one node from $X_i$ by definition and in the worst case, three is exactly one for each path and two additional neighbors.} (we call this weighted graph the {\em trace of the solution}), the number of partial solutions is bounded by $O(L^w.k^{(3w)^2})$ for each bag, which is constant for $k$ and $w$ bounded (we could obviously use better data structures to improve upon this value). We can now explain why it is enough to keep track of the  trace of each solution of $G_i$ for bag $X_i$ to build all partial solutions as we move up the tree decomposition. We need to distinguish according to the different types of bag.

\begin{itemize}

	\item \textbf{Leaf bag:}  In a leaf bag $X_i$, we have $|X_i|=\{v\}$ for some $v\in V$ and we can enumerate all $k$-path $L$-coloring of $G_i$ by enumerating all possible coloring of $v$. The partial solutions coincide with the solutions for $G_i$ so the trace is trivial in this case: it consists in the graph with vertex set $\{v\}$ and edge set $\{\}$, and any weight function as the set of edges is empty.
	\item \textbf{Forget bag:} In a forget bag $X_i$, we delete a vertex $v$ from a child bag $X_j$ i.e. $X_i = X_j \setminus \{v\}$ for some $v\in V$. In such a node of the tree decomposition, any  partial solution $P_j$ for $X_j$ yields a  partial solution $P_i$ for $X_i$. Indeed, any $k$-path $L$-coloring for $G_j$ that would be consistent with $P_j$ would again be a $k$-path $L$-coloring of $G_i$ as $G_i$ and $G_j$ coincide. The traces of the solution on $X_i$ and $X_j$ differ though, but we can easily recover the trace on $X_i$ from the trace on $X_j$.  Indeed, we can update the coloring and the trace as follows: the coloring is kept identical but it is restricted to the nodes of $X_i$ and the new trace is simply the trace on $X_i$ of the trace in $P_j$ (that is we simply `shrink' $v$). 
	\item \textbf{Introduce bag:} In an introduced bag $X_i$, we add a vertex $v$ from a child bag $X_j$ i.e. $X_i = X_j \cup \{v\}$ for some $v \in V$ (and $v \not \in V_j$). In order to check whether a $k$-path $L$-coloring $S_j$ of $G_j$ can be extended to a $k$-path $L$-coloring of $G_i$, we only need to check which coloring $c(v)$ of $v$ is compatible with $S_j$.  Adding $v$ to a color set adds edges in the subgraph of $G$ induced by the node of color $c(v)$.  We want the resulting graph to be a union of disjoint $F$-paths of $G_i$ of length at most $k$.  Let $G^{c(v)}_i$ (resp. $G^{c(v)}_j$) be the subgraph of $G_i$ (resp. $G_j$) induced by the node of color $c(v)$. Because of the structure of a tree decomposition, $v$ is only adjacent to vertices of $X_i$ in $G_i$. It follows that in order to check whether  $S_j$ can be extended, it is enough to check whether the set of edges ${\mathcal E}\subseteq E$ incident to both $v$ and a node of  color $c(v)$ in $X_i$ are all in $F$ and that adding $v$ and $\mathcal E$ (with weight one) to the trace $(H,f)$  of $G^{c(v)}_j$ on $X_j$ yields a weighted graph $({H',f'})$ whose support ${H'}$ is a union of disjoint paths of length at most $k$ with respect to ${f'}$ (note that adding $v$ might merge two previously disjoint paths). This can be checked in linear time by adapting for instance a depth first search algorithm. In case of a positive result,  substituting the trace $(H',f')$ by $(H,f)$ actually yields the trace on $X_i$ of the extension of $S_j$ to $G_i$.
	\item \textbf{Join bag:} In a join bag $X_i$, we want to check which partial solutions obtained for two different graphs $G_{j_1}$, $G_{j_2}$ associated with the two children bags $X_{j_1},X_{j_2}$ are ``compatible'', that is, would be the restriction of a $k$-path $L$-coloring of $G_i$. The logic is pretty similar to the previous situation. Let  $S_{j_1}$ and $S_{j_2}$ be two solutions for $G_{j_1}$ and $G_{j_2}$ respectively. We first need to check that common vertices (that is, vertices in $X_i$) are colored in the same way in both solutions. Then we need to check that the graph induced by each color set is a (disjoint) union of $F$-paths of length at most $k$. Because there is no edge between vertices in $V_{j_1}\setminus X_i$ and vertices in  $V_{j_2}\setminus X_i$ by the properties of the tree decomposition, it follows that the only obstruction can come from the fact that, in a color, the union of the edges of the disjoint $F$-paths in each solution (the union is indeed the induced graph) is not a union of disjoint $F$-paths of length at most $k$ .   We can restrict attention to the trace of the solutions on $X_i$ and we only need to check that the {\em union of the trace graphs} of same color induces a disjoint union of paths of length at most $k$. There is a subtlety though and we need to be careful about the definition of the union of the trace graphs. Here we mean the union of all paths from the trace graphs with the condition that two edges with same extremities are considered identical (and thus are not duplicated in the union) only if their weight is $1$ in both solutions (note that this can only happen to edges in $E[X_i]$ as edges of length one that do not connect two vertices in $X_i$ contain at least one vertex in $V_{j_1} \setminus X_i$ or $V_{j_2} \setminus X_i$ and can thus only appear in one of the two solutions). Indeed, otherwise they must be considered different as they correspond to pieces of paths of $G$ whose internal nodes (in $V_{j_1} \setminus X_i$ or $V_{j_2} \setminus X_i$) where shrunk, and the union graph should be seen as a multi-graph. This is to deal with the special case that, when an edge of the traces of same color corresponds to an edge of $G$ between two vertices of $X_i$, it can (and will) be part of both partial solutions for $G_{j_1}$ and $G_{j_2}$. In case of a positive outcome, the weighted graph whose (weighted) adjacency matrix $A_i$ is $\max(A_{j_1}, A_{j_2})$ (where  $A_{j_1}$ and $A_{j_2}$ are the (weighted) adjacency matrices of the traces in the partial solutions for $G_{j_1}$ and $G_{j_2}$ respectively) yields the trace on $X_i$ of the associated $k$-path $L$-coloring of $G_i$. 
	
\end{itemize} 

The discussion above shows that we do not miss any partial solution as we move up the tree (and that each partial solution we build is actually valid). Hence there exist a $k$-path coloring of $G$ if and only if there exists a partial solution in the root node of the tree decomposition when applying the procedure above. 

The computational time at each node of the nice tree decomposition is dominated by the join bag case. For each coloring of the vertices of $X_i$ we then need to check whether the union of the support graph of the trace of each possible partial coloring for $X_{j_{1}}$ and $X_{j_{2}}$ are compatible: in the worst case we need to check $L^w$ coloring, and $k^{(3w)^2} \times k^{(3w)^2} $ pairs of traces, and checking that the union of the traces in each color still yields a union of disjoint paths of length at most $k$ can be done in time $O((3w)^2)$, by first building  the weighted adjacency matrix of the union of the trace (and checking that it does not contain multi-edges) and by then adapting the depth first search algorithm, since each weighted adjacency matrix has size at most $(3w)^2$. The overall complexity is thus bounded by $O(L^w.k^{2(3w)^2}.(3w)^2.|T|)$.


\end{proof}

\section{Numerical Experiments}

In ~\cite{Arnborg_1987}, Arnborg et al have proved that deciding if the treewidth of a graph $G$ is at most $w$, where $w\geq 0$, is NP-complete. However, there are good heuristics to determine a tree decomposition of a given graph $G$ with a width `close to' the treewidth. Different heuristics are presented and compared in \cite{Bodlaender3}. Moreover, Arnborg et al ~\cite{Arnborg_1987} showed that for every fixed value of $w$, there is a linear-time algorithm that finds a tree decomposition of width $w$ (if it exists). 

We propose to solve DSA-MP on real instances arising from integrated circuit manufacturing by first using a heuristic to get a `small' tree-decomposition of the graph, and by then using the dynamic programming algorithm described in the previous section to solve the problem (actually, we tailored the algorithms to the case where $k=1$ and $k=2$ to make them slightly simpler to implement, see \cite{Dehia} for the details about the corresponding implementations). We used D-FLAT to implement the corresponding algorithm~\cite{DFLAT_2014}. D-FLAT has the advantage to implement different state-of-the-art heuristics to find a close-to-optimal nice tree decompositions and offers a generic langage to describe how to extend solutions for each type of nodes of the nice tree decomposition. We ran D-FLAT iteratively to solve the $k$-path $L$-colorability problem  starting from $L=2$ and increasing $L$ until a solution was found. All tests were done on a machine equiped with an Intel(R) Xeon(R) CPU E5-2640 2.60 GHz and a memory of 529GB. As already observed, the typical industrial designs are made of several hundred thousands of vias, but because the placement of the vias is made as to anticipate as much as possible conflicts that may arise from lithography, the `conflict graph' is usually extremely sparse and contains only hundreds or thousands of different connected components. Since the optimization of each connected component can be parallelized, we focus attention on the computational time for each connected component individually.  

We report hereafter computational experiments on 23 connected components of increasing size arising from true industrial instances in Table \ref{results}. We can see that the linear time complexity is confirmed experimentally.

\begin{table}[h!]
\begin{center}
\scalebox{0.8}{\begin{tabular}{|c|c|c|c|c|c|c|c|c|c|c|}
\hline
Instance name  & $|V|$ & $|E|$ & $|F|$ &  $\omega(G)$ & $\Delta(G)$ & DFLAT\_TW(G) &  $\chi_{path}^{1}$  & cpu time (sec) & $\chi_{path}^2$ & cpu time (sec) \\
\hline
industrial\_1	& 54	& 56	& 52	&   3	& 4	& 2 &	2 &	0.35	&	2	&	0.92 \\
industrial\_2	& 61	& 86	& 85	&   3	& 5	& 2 &	2 &	0.75	&	2	&	1.62 \\ 
industrial\_3	& 89	& 112	& 107	&   3	& 5	& 2 &	2 &	0.51	&	2	&	1.5 \\
industrial\_4	& 90	& 113	& 109	&   3	& 5	& 2 &	2 &	0.98	&	2	&	1.21 \\
industrial\_5	& 96	& 111	& 105	&   3	& 4	& 2 &	2 &	0.55	&	2	&	1.73 \\
industrial\_6	& 98	& 126	& 120	&   3	& 5	& 3 &	2 &	0.84	&	2	&	3.45 \\
industrial\_7	& 102	& 144	& 141	&   3	& 5	& 2 & 2 &	0.94	&	2	&	1.47 \\
industrial\_8	& 111	& 140	& 137	&   3	& 5	& 2 &	2 &	1.45	&	2	&	2.21 \\
industrial\_9	& 114	& 131	& 125	&   3	& 4	& 2 &	2 &	0.79	&	2	&	2.21 \\
industrial\_10	& 116	& 155	& 151	&   3	& 5	& 3 &	2 &	1.15	&	2	&	2.5 \\
industrial\_11	& 119	& 142	& 136	&   3	& 4	& 3 &	2 &	1.41	&	2	&	4.02 \\
industrial\_12	& 128	& 159	& 149	&   3 & 5	& 2 &	2 &	1.17	&	2	&	2.83\\
industrial\_13	& 137	& 167	& 160	&   3	& 5	& 3 &	2 &	1.22	&	2	&	3.05\\
industrial\_14	& 159	& 196	& 188	&   3	& 5	& 2 &	2 &	1.42	&	2	&	3.53\\
industrial\_15	& 173	& 224	& 216	&   3	& 5	& 3 &	2 &	1.83	&	2	&	3.44\\
industrial\_16	& 382	& 396	& 339	&   3	& 4	& 2 &	2 &	2.57	&	2	&	4.95\\
industrial\_17	& 969	& 1001	& 900	&   3	& 3	& 2 &	2 &	7.26	&	2	&	12.52\\
industrial\_18	& 993	& 1009	& 927	&   3	& 4	& 2 &	2 &	6.19	&	2	&	12.63\\
industrial\_19	& 997	& 1047	& 906	&   3	& 4	& 2 &	3 &	8.86	&	2	&	12.82\\
industrial\_20	& 998	& 1024	& 924	&   3	& 4	& 2 &	3 &	8.65	&	2	&	12.95\\
industrial\_21	& 1900	& 1937	& 1804	&   3	& 4	& 2 &	3 &	21.09	&	2	&	26.87\\
industrial\_22	& 1912	& 1960	& 1809	&   3	& 4	& 2 &	3 &	18.22	&	2	&	26\\
industrial\_23	& 1937	& 1996	& 1812	&   3	& 4	& 2 &	2 &	6.29	&	2	&	27.02 \\
\hline
\end{tabular}}
\caption{Industrial instances characteristics and results: $\omega(G)$ is the size of the maximum clique in $G$,  $\Delta(G)$ is the maximum degree of $G$, DFLAT\_TW(G) is the width of the tree decomposition returned by D-FLAT heuristics, and $\chi_{path}^{k}$ is the $k$-path chromatic number.}\label{results}
\end{center}
\end{table}

All instances could be solved in less than 30 seconds. The 23 industrial instances used in this study share similar properties with the pseudo-industrial instances generated in \cite{IP}, when the resolution limit is set to 31nm.  We thus compared the two approaches on the same set of instances (and on the same machine) and the results are reported in  Table \ref{results2}.  

\begin{table}[h!]
\begin{center}
\scalebox{0.65}{\begin{tabular}{|l|l|l|l|l|l|l|l|l|l|l|l|l|}
\hline
 Instance & $|V| $	&$|E| $	&$|F| $ &$\omega(G)$	&$\Delta(G)$ & DFLAT\_TW(G) &  $\chi_{path}^{1}$  & DFLAT time 1 (sec) & IP 1 time (sec) & $\chi_{path}^2$ & DFLAT time 2(sec) & IP 2 time (sec)  \\
 \hline
	clip1\_31&	191& 	242&		242	& 	 	3& 	5  &3		&2	&2,71 	&0,8		&2	&465,74	&1,93\\
	clip2\_31&	139&	 	188&		188	& 	 	3& 	5  &3 	&3	&21,15 	&1,54	&2	&64,89	&2,11\\
	clip3\_31&	98&		117&		108	& 	 	3& 	4  &2	 	&2	&0,8		&0,11	&2	&1,69	&0,34\\
	clip4\_31&	120&		147&		139	& 	 	3& 	4 &3		&2	&0,87	&0,49	&2	&9,45	&0,46\\
	clip5\_31&	170&		213& 	213	& 	 	3& 	5 &3		&3	&5,45	&0,99	&2	&9,57	&1,52\\
	clip6\_31&	178&		229& 	229	& 	 	3& 	5 &3		&2	&2,1		&0,68	&2	&38,45	&1,89\\
	clip7\_31&	203&		256& 	223	& 	 	3& 	5 &3		&3	&8,83	&0,94	&2	&92,05	&0,73\\
	clip8\_31&	122&		162& 	160	& 	 	3& 	5 &3		&2	&0,83	&0,45	&2	&2,89	&1,35\\
	clip9\_31&	152&		193& 	193	& 	 	3& 	4 &3		&2	&2,8		&0,28	&2	&31,13	&1,96\\
	clip10\_31&139&	175& 	175	& 	 	4& 	4 &3		&2	&0,91	&0,51	&2	&2,3		&0,64\\
\hline
\end{tabular}}
\caption{Pseudo-industrial instances characteristics and results: $\omega(G)$ is the size of the maximum clique in $G$,  $\Delta(G)$ is the maximum degree of $G$, DFLAT\_TW(G) is the width of the tree decomposition returned by D-FLAT heuristics, DFLAT k time (resp. IP k time) represents the time used by DFLAT (resp. IP) to solve the k-path coloring problem,  and $\chi_{path}^{k}$ is the $k$-path chromatic number.}\label{results2}
\end{center}
\end{table}

It appears that the best integer programming formulations from  \cite{IP} outperforms the dynamic programming approach on these instances. It is not completely clear to us whether building upon heuristics that would be tailored to the industrial setting to find a small tree decomposition and developing a finer implementation of our dynamic programming approach (with possibly some improvement in the data structures to speed up the algorithm) could lead to substantial computational improvements and could make the technique compete (or even beat) the results obtained with IP on the typical industrial instances.  However, the results obtained for the larger pseudo-industrial instances used in \cite{IP} (when the resolution limit is set to 39nm and 49nm) did not encourage us to pursue this line of research:
\begin{itemize}
\item For instances with tree-width four, the computational times were still ``reasonable'' to hope that a better implementation could help make the technique competitive: the computation time ranges from several minutes to one hour in this case, while the IP formulation can  solve all instances within seconds or minutes (note that the size of the instances, i.e. the largest connected component, is now in the interval $[816,3850]$).  
\item For instances with tree-width five or more, the computation time are becoming prohibitive though compared to IP : several hours or more for dynamic programming versus a few minutes for IP. 
\end{itemize} 
Even though we could also use additional `tricks' to reduce the size of the instances (for instance, if an edge of $E\setminus F$ disconnects some connected component, the problem can be solved independently on both subgraphs and the solutions recombined later  - if the $k$-path chromatic number is at least two), we did not investigate this direction further, as the same techniques could also be exploited by the IP model, and this is this later direction that is currently investigated by Mentor Graphics to see whether IP can be made competitive with their in-house heuristics, that can solve all instances in less than a second (the corresponding `tricks' are already implemented in Mentor Graphics' heuristics).

\section{Acknowledgment}

This project has been partly supported by the Association Nationale de la Recherche et de la Technologie (Convention CIFRE 2015/0553).

\bibliographystyle{abbrv}
\bibliography{bib_graph,bib_semic}

\end{document}